\documentclass[english,3p,twocolumn,superscriptaddress]{revtex4-1}
\usepackage[latin9]{inputenc}
\usepackage{amsmath}
\usepackage{amssymb, enumerate}
\usepackage{color}
\usepackage{babel}
\usepackage{graphicx}
\usepackage{epstopdf}
\usepackage{float}
\usepackage{thmtools}
\usepackage{thm-restate}
\usepackage{hyperref}
\usepackage{cleveref}
\usepackage{enumerate}
\usepackage{lipsum}
\usepackage{fixmath}
\usepackage{float}
\makeatletter

\@ifundefined{textcolor}{}
{%
 \definecolor{BLACK}{gray}{0}
 \definecolor{WHITE}{gray}{1}
 \definecolor{RED}{rgb}{1,0,0}
 \definecolor{GREEN}{rgb}{0,1,0}
 \definecolor{BLUE}{rgb}{0,0,1}
 \definecolor{CYAN}{cmyk}{1,0,0,0}
 \definecolor{MAGENTA}{cmyk}{0,1,0,0}
 \definecolor{YELLOW}{cmyk}{0,0,1,0}
}

\newtheorem{theorem}{Theorem}

\newtheorem{proposition}[theorem]{Proposition}
\newtheorem{example}[theorem]{Example}
\newtheorem{remark}[theorem]{Remark}

\newtheorem{definition}[theorem]{Definition}

\newenvironment{proof}[1][Proof]{\noindent\textbf{#1.} }{\ \rule{0.5em}{0.5em}}

\makeatother

\begin{document}

\title{Information backflow may not indicate quantum memory}

\author{Micha{\l} Banacki}
\affiliation{International Centre for Theory of Quantum Technologies, University of Gda\'{n}sk, 80-308 Gda\'{n}sk, Poland}
\affiliation{Institute of Theoretical Physics and Astrophysics, Faculty of Mathematics, Physics and Informatics, University of Gda\'{n}sk, 80-308 Gda\'{n}sk, Poland}
\author{Marcin Marciniak}
\affiliation{Institute of Theoretical Physics and Astrophysics, Faculty of Mathematics, Physics and Informatics, University of Gda\'{n}sk, 80-308 Gda\'{n}sk, Poland}
\author{Karol Horodecki}
\affiliation{International Centre for Theory of Quantum Technologies, University of Gda\'{n}sk, 80-308 Gda\'{n}sk, Poland}
\affiliation{Institute of Informatics, Faculty of Mathematics, Physics and Informatics, National Quantum Information Centre, University of Gda\'{n}sk, 80-308 Gda\'{n}sk, Poland}
\author{Pawe{\l} Horodecki}
\affiliation{International Centre for Theory of Quantum Technologies, University of Gda\'{n}sk, 80-308 Gda\'{n}sk, Poland}
\affiliation{Faculty of Applied Physics and Mathematics, National Quantum Information Centre, Gda\'{n}sk University of Technology, 80-233 Gda\'{n}sk, Poland}

\begin{abstract}We analyze recent approaches to quantum Markovianity and how they relate to the proper definition of quantum memory. We point out that the well-known criterion of information backflow may not correctly report character of the memory falsely signaling its quantumness. Therefore, as a complement to the well-known criteria, we propose several concepts of {\it elementary dynamical maps}. Maps of this type do not increase distinguishability of states which are indistinguishable by von Neumann measurements in a given basis. Those notions and convexity allows us to define general classes of processes without quantum memory in a weak and strong sense. Finally, we provide a practical characterization of the most intuitive class in terms of the new concept of {\it witness of quantum information backflow}.
\end{abstract}



\maketitle


\textit{Introduction.-} Nowadays, due to constant development in both theoretical and experimental branches of quantum information theory, topic of quantum memory become more and more relevant. In particular idea of Markovian evolution (evolution without memory) coming form the theory of quantum open systems \cite{RH12,BP07} has been recently studied in extensive way within different frameworks \cite{BLPV16,RHP14,PRFPM18}. While well-defined for classical case, notion of Markovianity is still lacking a single definition in a quantum setting. Two typical approaches used for description of memoryless processes are either related to the divisibility of dynamics \cite{RHP10} or the backflow of information \cite{BLP09}.

The main interest related to Markovianity comes from negation of its definition - processes described as non-Markovian should express quantum memory effect. Quantum non-Markovinity can be therefore treated as a resource for various quantum informational tasks like quantum computation, communication or cryptography. Despite this philosophy, current approaches to Markovianity in a quantum setting usually do not concern themselves with proper distinction between memory effects which are truly quantum. The main idea behind this paper is to challenge this status quo and show different direction for further investigation.

The aim of this note is to propose a new approach towards description of quantum information blackflow. In section \textit{Preliminaries on Markovianity} we briefly recall notions of CP-divisibility and lack of information backflow, when in section \textit{Backflow of a classical type} we discuss realization of classical dynamics as quantum dynamical maps and introduce example indicating conceptual problem with current concept of quantum memory.

In section \textit{Elementary dynamical maps} we present a definition of elementary dynamical maps, while in section \textit{Quantum memory} we use this notion to propose generalized concept of dynamical maps without quantum memory (i.e. with no quantum information backflow). Finally we formulate a few open questions within presented framework.

\textit{Preliminaries on Markovianity.-} Consider a finite-dimensional quantum system. General form of evolution of that system can be given by time-dependent family of quantum channels or completely positive and trace preserving (CPTP) maps $\Lambda_t:M_d(\mathbb{C})\rightarrow M_d(\mathbb{C})$. We will call such family by the name of dynamical map and denote it by $\left\{\Lambda_t\right\}_{t}$ with $t$ belonging to some considered interval of observation.

Notion of Markovianity is related to processes without a memory. In the classical case this idea is unambiguously defined. However, on the contrary, there in no definition of quantum Markovianity which can be commonly seen as appropriate. There are at least two approaches in which one can define memoryless quantum dynamical map.

The first definition was proposed by Rivas, Huelga and Plenio \cite{RHP10} and it is based on the concept of divisibility of dynamics.

\begin{definition}\label{RHP}
Consider a dynamical map $\left\{\Lambda_t\right\}_t$ acting between $M_d(\mathbb{C})$. We say that $\left\{\Lambda_t\right\}_t$  is CP-divisible or RHP-Markovian if for all $t,s$ such that $t\geq s$ there exists a CPTP map $V_{t,s}$ (so called propagator) for which
\begin{equation}\nonumber
\Lambda_t=V_{t,s}\Lambda_s.
\end{equation}
\end{definition}Such approach is justify by memoryless environment interpretation \cite{CRS18}. In particular any dynamical map arising from differential equation
\begin{equation}\label{GKLS}
\frac{d\Lambda_t}{dt}=\mathcal{L}_t\Lambda_t, \ \Lambda_0=\mathrm{Id}
\end{equation}with a time-dependent Gorini-Kossakowski-Lindblad-Sudarshan (GKLS) generator $\mathcal{L}_t$ \cite{GKS76,L76} is CP-divisible.

The second broadly accepted approach to quantum Markovianity (postulated by Breuer, Laine and Piilo \cite{BLP09}) is based on the idea of information backflow in which increase in distinguishability of states is an indicator of a flow of information form environment to the system and is interpreted as a quantum memory effect.
\begin{definition}\label{BLP}
Consider a dynamical map $\left\{\Lambda_t\right\}_t$ acting between $M_d(\mathbb{C})$. We say that $\left\{\Lambda_t\right\}_t$ is BLP-Markovian (i.e. there is no information backflow) if the following
\begin{equation}\nonumber
\frac{d}{dt}\left\|\Lambda_t(\rho_1-\rho_2)\right\|_1\leq 0
\end{equation}holds for any pair of states $\rho_1,\rho_2\in M_d(\mathbb{C})_+$ (here $\left\|.\right\|_1$ denotes the trace norm).
\end{definition}Lack of backflow has a clear operational meaning - it states that there is no memory effect in a quantum process if the probability that two states (after evolution according to dynamical map) can be distinguished is not increasing in time.

Since trace preserving map is positive if and only if it is a contraction with respect to the trace norm \cite{P03}, CP-divisibility implies no backflow of information, but the converse implication is generally not true, even in the case of classical dynamics \cite{CKR11} (notice that also stronger version of BLP-Markovianity and its relation with divisibility was discussed in the literature \cite{CKR11,CM14}). Therefore, we adapt definition \ref{BLP} and the concept of information backflow as a fundamental notion and starting point for our further considerations. 

\textit{Backflow of a classical type.-} Let us consider a probability vector $\vec{p}=(p_1, p_2,\ldots ,p_d)$ describing some classical d-level system. Classical stochastic dynamics in that setting is given by the time-dependent ($t\geq 0$) stochastic matrix $\Lambda(t)$ (i.e. $\Lambda(t)\in M_d(\mathcal{C})$, $\Lambda(t)_{ij}\geq 0$ for all $i,j$ and $\sum_{i=1}^d\Lambda_{ij}(t)=1$ for all $j$) via equation $\vec{p}(t)=\Lambda(t)\vec{p}$. Notice that choosing a particular orthonormal basis $B$ in $\mathbb{C}^d$, one can treat quantum state $\rho=\sum_i^dp_i\left|e_i\right\rangle\left\langle e_i\right|$ diagonal with respect to that basis as a probability vector $\vec{p}$ describing classical system. If so then classical dynamics described above can be expressed in the quantum dynamical setting \cite{CKR11} by a dyamical map given as
\begin{equation}\label{map_classical}
\Lambda^{cl}_t(\rho)=\sum_{i,j=1}^d \Lambda_{ij}(t)\left|e_i\right\rangle\left\langle e_j\right|\rho\left|e_j\right\rangle\left\langle e_i\right|
\end{equation}where $\Lambda(t)$ is some time-dependent stochastic matrix. Allowing for action of time-dependent unitary we obtain a generalization of classical dynamics
\begin{equation}\label{map_2}
\Lambda^{gcl}_t=U(t)\Lambda^{cl}_tU^{\dagger}(t). 
\end{equation}which up to unitary may be seen as a classical stochastic dynamics.

\begin{example}\label{3classical}
Consider a time-dependent family of qubit depolarizing channels
\begin{equation}\label{dep}
\Lambda_t(\rho)=\lambda(t)\rho+\frac{1-\lambda(t)}{2}\mathbb{I}
\end{equation}with 
\begin{equation}\nonumber
\lambda(t)= \begin{cases} \frac{4+\epsilon}{6t_0^2}\left(t-t_0\right)^2+\frac{2-\epsilon}{6} &\mbox{if } \ 0\leq t< t_0 \\ 
\frac{1}{6}+\frac{1-\epsilon}{6}\cos (t-t_0) & \mbox{if }\ t_0\leq t \end{cases}
\end{equation}where $\epsilon$ is some arbitrarily small positive number. 

Observe that (\ref{dep}) is governed by equation \ref{GKLS} with $\mathcal{L}_t=\sum_{i=1}^3\gamma(t)\frac{1}{2}\left(\sigma_i\rho\sigma_i-\rho\right)$. Because $\gamma(t)\geq 0$ for $\left[0,t_0\right)$, discussed dynamics is CP-divisible on that interval. Notice that on the contrary, for $\left[t_0,\infty\right)$ one can observe information backflow. On the other hand, for any $t\in \left[t_0,\infty\right)$, $\Lambda_t$ can be expressed by a time-independent convex combination
$\Lambda_t=\frac{1}{3}\left(\Lambda_t^{(1)}+\Lambda_t^{(2)}+\Lambda_t^{(3)}\right)$ where any $\Lambda_t^{(k)}$ is classical dynamical map (with respect to the basis consisting of eigenvectors of $\sigma_k$) of the form \ref{map_classical}, governed by the same bistochastic time-dependent matrix $\Lambda(t)$ such that $\Lambda_{11}(t)=\frac{3}{4}+\frac{1-\epsilon}{4}\cos\left(t-t_0\right)$.
\end{example}

Notice that in the case of previous example, any possible backflow of information can be simulated by classical dynamics which has been chosen with some probability. In particular observed backflow should not be treated as a quantum phenomenon.

\textit{Elementary dynamical maps.-} Discussion in the previous section shows that one should be careful while talking about memory and its character in quantum dynamical setting. 
Somewhat similar objections towards presence of memory effects for dynamical maps which are not CP-divisible (but have no information backflow) were described in \cite{MCPS17}. Here we want to address issue of classicality of observed memory (see also recent results in \cite{METTPSH19}).

Consider two different states $\rho_1, \rho_2\in M_d(\mathbb{C})_{+}$ and assume that they are distinguishable with respect to von Neumann measurements (i.e. projective measurements) in some fixed orthonormal basis $B$ in $\mathbb{C}^d$. Difference in diagonals of this states can be seen as an information encoded in a classical basis, therefore possible increase in distinguishability of such pair of states during some evolution should be treated as an effect of classical memory. This motivates the following definition.
\begin{definition}\label{elementary_qm}
Let $B$ be an orthonormal basis in $\mathbb{C}^d$. Consider a dynamical map $\left\{\Lambda_t\right\}_t$ acting between $M_d(\mathbb{C})$. We say that $\Lambda_t$ is elementary (with respect to $B$) if 
\begin{equation}\nonumber
\frac{d}{dt}\left\|\Lambda_t(\rho_1-\rho_2)\right\|_1\leq 0
\end{equation}for any pair of states $\rho_1,\rho_2\in M_d(\mathbb{C})_+$ that are indistinguishable by von Neumann measurements in $B$.
\end{definition}

Recall that a CPTP map $\Omega$ belongs to the set of dephasing-covariant incoherent
operations (DIO) related to given basis $B$ \cite{SAP17,CG16,MS16}, when it commutes with the dephasing map, i.e. $\Omega(\Delta(\rho))=\Delta(\Omega(\rho))$ where $\Delta(\rho)=\sum_i\left|e_i\right\rangle\left\langle e_i\right|\rho\left|e_i\right\rangle\left\langle e_i\right|$.

\begin{proposition}Let dynamical map $\left\{\Lambda_t\right\}_t$ be elementary with respect to basis $B$. If $\Omega$ is DIO with respect to the same basis, then $\left\{\Lambda_t\circ \Omega\right\}_t$ is elementary as well.
\end{proposition}
\begin{proof}Let $\rho_1,\rho_2$ be a pair of states indistinguishable by von Neumann measurements in basis $B$. Then $\Omega(\Delta(\rho_1-\rho_2))=0=\Delta(\Omega(\rho_1-\rho_2))$. Therefore $\sigma_i=\Omega(\rho_i)$ for $i=1,2$ are once more indistinguishable by von Neumann measurements in basis $B$ from which follows 
\begin{equation}\nonumber
\frac{d}{dt}\left\|(\Lambda_t\circ \Omega)(\rho_1-\rho_2)\right\|_1=\frac{d}{dt}\left\|\Lambda_t(\sigma_1-\sigma_2)\right\|_1\leq 0.
\end{equation}
\end{proof}

Let us now reformulate condition from definition \ref{elementary_qm} in the case of a qubit dynamics. Any CPTP map $\Lambda:M_2(\mathbb{C})\rightarrow M_2(\mathbb{C})$ can be described by a matrix (expressed with respect to the basis $\left\{\mathbb{I},\sigma_1,\sigma_2,\sigma_3\right\}$)
\begin{equation}\label{form0}
\Lambda=\begin{bmatrix}
    1   & \vec{0} \\
    \vec{r}  &T 
\end{bmatrix}
\end{equation}with real vector $\vec{r}$ and real $3$ by $3$ matrix $T$ fulfilling additional conditions \cite{RSW02} assuring complete positivity of $\Lambda$. Dynamical map $\left\{\Lambda_t\right\}_t$ for a qubit system is then given by a time-dependent family of matrices (\ref{form0}). Observe that $\left\|\Lambda_t(\rho_1-\rho_2)\right\|_1=|T(t)(\vec{m}_1-\vec{m}_2)|$ for any two states $\rho_i=\frac{1}{2}(\mathbb{I}+\vec{m_i}\vec{\sigma})$. Therefore, dynamical map $\left\{\Lambda_t\right\}_t$ is elementary with respect to the basis of eigenvectors of $\vec{n}\vec{\sigma}$ (any basis is of this form) if and only if 
\begin{equation}\label{condition0}
\forall_{\vec{m}\perp\vec{n}}\ \ \frac{d}{dt}|T(t)\vec{m}|\leq 0.
\end{equation}Note that condition (\ref{condition0}) depends only on the $T(t)$ part of matrix (\ref{form0}). Observe that in condition (\ref{condition0}) one can exchange $|T(t)\vec{m}|$ with its square and derive equivalent formula 
\begin{equation}
\forall_{\vec{m}\perp \vec{n}}\ \  \left(\vec{m},X(t)\vec{m}\right) \leq 0
\end{equation}where $X(t)=\frac{d}{dt}T^T(t)T(t)$.

Now choose orthogonal matrix $O$ such that $O\vec{m}=(m_1,m_2,0)^T=\vec{m}_I$ for any $\vec{m}\perp\vec{n}$ and denote $T_I(t)=T(t)O^T$. Observe that $T_{II}(t)\vec{m}_I=T(t)\vec{m}$ where $T_{II}(t)$ is equal to $T_I(t)$ with last column exchange with zero vector. By SDV decomposition with two orthogonal matrices $O_I, O_{II}$ we obtain $T_{II}(t)=O_I(t)\tilde{T}(t)O_{II}(t)$ where
\begin{equation}\nonumber
\tilde{T}(t)=\begin{bmatrix}
    \lambda_1(t)  &0   & 0   \\ 
		0 & \lambda_2(t) & 0              \\ 
		0 &0 &0
\end{bmatrix}, \ O_{II}(t)=\begin{bmatrix}
    o_{11}(t)  &o_{12}(t)  & 0   \\ 
		o_{21}(t) & o_{22}(t) & 0              \\ 
		0 &0 &1
\end{bmatrix}.
\end{equation}Since $O_{I}(t)$ does not change norm of a vector, then by (\ref{condition0}) the necessary condition for $T(t)$ to describe elementary dynamical map is given by inequality
\begin{equation}\label{ineq1}
\mathrm{max}\left\{\lambda_1(t),\lambda_2(t)\right\}\leq \mathrm{max}\left\{\lambda_1(s),\lambda_2(s)\right\}
\end{equation}satisfied for any $t>s$, while the sufficient one is described by inequality
\begin{equation}\label{ineq2}
\mathrm{max}\left\{\lambda_1(t),\lambda_2(t)\right\}\leq \mathrm{min}\left\{\lambda_1(s),\lambda_2(s)\right\}
\end{equation}satisfied for any $t>s$. When (\ref{ineq1}) is satisfied but (\ref{ineq2}) is not, $T(t)$ may still describe elementary dynamical map depending on how fast $O_{II}(t)$ change with time in comparison to change of $\lambda_1(t)$ and $\lambda_2(t)$.

Let $S(t)$ describes a notrivial part of $\tilde{T}(t)O_{II}(t)$, i.e.
\begin{equation}\nonumber
S(t)=\begin{bmatrix}
   \lambda_1(t) o_{11}(t)  & \lambda_1(t) o_{12}(t)\\ 
		\lambda_2(t) o_{21}(t) & \lambda_2(t)o_{22}(t)
\end{bmatrix}.
\end{equation}Condition (\ref{condition0}) may be reformulated in the following form
\begin{equation}\nonumber
\forall_{\vec{m}(\alpha)}\ \ \frac{d}{dt}|S(t)\vec{m}(\alpha)|^2\leq 0
\end{equation}where $\vec{m}(\alpha)=(\cos\alpha,\sin\alpha)^T$. By direct differentiation this can be equivalently stated as  
\begin{equation}\label{opti}
\forall_{\vec{m}(\alpha)}\ \  \left(\vec{m}(\alpha),\tilde{S}(t)\vec{m}(\alpha)\right) \leq 0
\end{equation}with $\tilde{S}(t)=\frac{d}{dt}S^T(t)S(t)$. Maximization of (\ref{opti}) over angle $\alpha$ leads to 
\begin{equation}\label{condition2}
\tilde{S}_{11}+\tilde{S}_{22}+\sqrt{(\tilde{S}_{11}-\tilde{S}_{22})^2+(\tilde{S}_{12}+\tilde{S}_{21})^2}\leq 0
\end{equation}where all $\tilde{S}_{ij}=\tilde{S}_{ij}(t)$ are time-dependent. Therefore, we arrive at full characterization of qubit elementary dynamical maps.

\begin{proposition}Qubit dynamical map $\left\{\Lambda_t\right\}_t$ is elementary with respect to the basis of eigenvectors of $\vec{n}\vec{\sigma}$ if and only if $\tilde{S}(t)$ satisfies condition (\ref{condition2})
\end{proposition}

In a search for appropriate notion of lack of quantum memory, one may consider other subclasses of elementary dynamical maps which more and more resemble classical dynamics (see schematic figure \ref{fig1}).
\begin{definition}\label{def2}Let $B$ be an orthonormal basis in $\mathbb{C}^d$. Consider a dynamical map $\left\{\Lambda_t\right\}_t$ acting between $M_d(\mathbb{C})$. We say that $\left\{\Lambda_t\right\}_t$ is block-diagonal elementary (with respect to $B$) if it is elementary and there exist a time-dependent unitary $U(t)$ such that $U(t)\Lambda_t U(t)^{\dagger}$ is block-diagonal with respect to the direct sum decomposition
\begin{equation}\nonumber
M_d(\mathbb{C})=M_d(\mathbb{C})_{\mathrm{diag}}\bigoplus M_d(\mathbb{C})_{\mathrm{off}}
\end{equation}where subscript diag and off refers to diagonal and off-diagonal part (for a given basis $B$) respectively. If $U(t)$ is such that only 
\begin{equation}\nonumber
U(t)\Lambda_t\left(M_d(\mathbb{C})_{\mathrm{diag}}\right)U(t)^{\dagger}\subset M_d(\mathbb{C})_{\mathrm{diag}}
\end{equation}is satisfied, then we say that $\left\{\Lambda_t\right\}_t$ is diagonal elementary (with respect to $B$).
\end{definition}

In a simplest qubit case dynamical maps of this type can be completely characterized by coherence properties.
\begin{proposition}
Let $B$ be an orthonormal basis in $\mathbb{C}^2$ and let $\left\{\Lambda_t\right\}_t$ be a qubit dynamical map such that $U(t)\Lambda_t(\rho)U(t)^{\dagger}$ is block-diagonal for some unitary $U(t)$. Then $\left\{\Lambda_t\right\}_t$ is elementary if and only if a $l_1$-coherence measure $C_{l_1}(U(t)\Lambda_t(\rho)U(t)^{\dagger})$ is time non-increasing for any state $\rho \in M_2(\mathbb{C})_{+}$.
\end{proposition}
\begin{proof}Let $\tilde{\Lambda}_t=U(t)\Lambda_tU(t)^{\dagger}$. Note that for qubits $C_{l_1}(\rho)=\sum_{i\neq j}\left|\rho_{ij}\right|=\inf_{\sigma\in \mathrm{diag}(B)}\left\|\rho-\sigma\right\|_1=\left\|\rho_{\mathrm{off}}\right\|_1$ and 
\begin{eqnarray}
  C\left(\tilde{\Lambda}_t(\rho)\right) &=& \left\|\tilde{\Lambda}_t(\rho_{\mathrm{off}})_{\mathrm{off}}+\tilde{\Lambda}_t(\rho_{\mathrm{diag}})_{\mathrm{off}}\right\|_1 \nonumber \\
   &=& \left\|\tilde{\Lambda}_t(\rho_{\mathrm{off}})\right\|_1=\left\|\alpha\tilde{\Lambda}_t(\rho_1-\rho_2)\right\|_1 \nonumber \\
	&=&\left\|\alpha\Lambda_t(\rho_1-\rho_2)\right\|_1 \nonumber 
\end{eqnarray}where $\alpha\geq 0$ and $\rho_1, \rho_2$ are some states indistinguishable by von Neumann measurements in basis $B$.
\end{proof}\newline

\begin{figure}[H]
\includegraphics[width=0.45\textwidth]{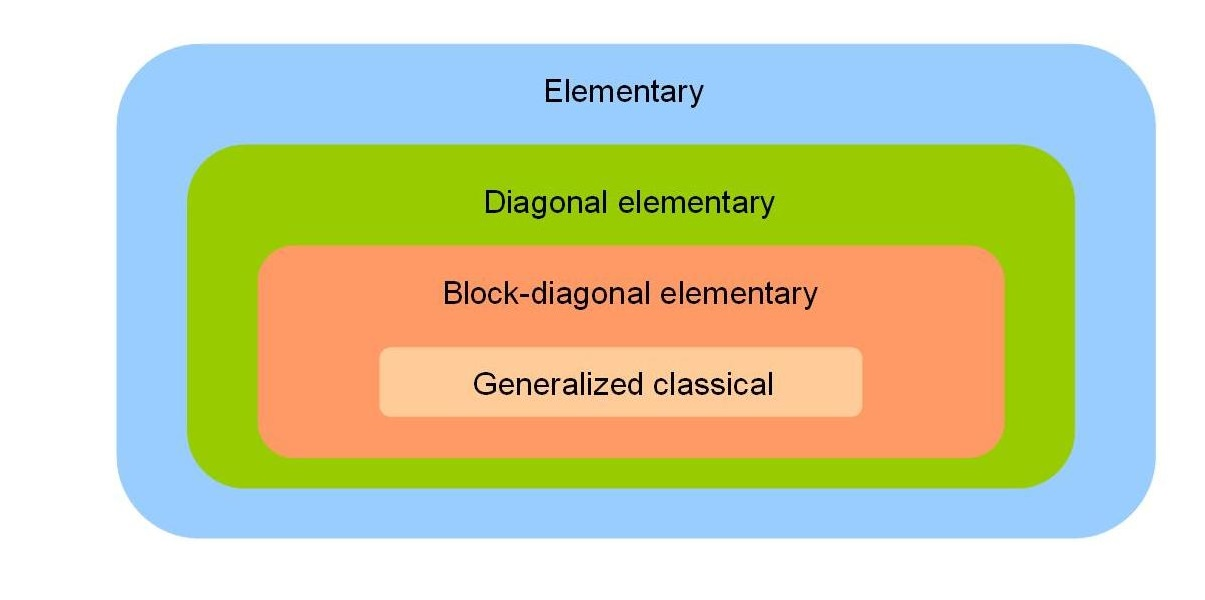}
\caption{\label{fig1}Relations between classes of dynamical maps for a given basis $B$.}
\end{figure}

Finally, observe that trivially all generalized classical (\ref{map_2}) dynamical maps are block-diagonal elementary (see schematic figure \ref{fig1}). Notice that generalized classical dynamics (\ref{map_2}) and dynamical maps from definition \ref{def2} are expressed up to a local unitary as this operation can always be applied locally to the system and it is not related to environment.

\textit{Quantum memory.-} In order to further generalize content of definition \ref{elementary_qm}, examine the following example.

\begin{example}\label{pauli}Consider a qubit dynamical map 
\begin{equation}
\Lambda_t(\rho)=\frac{1}{2}(\mathbb{I}+\lambda_1(t)n_1\sigma_1+\lambda_2(t)n_2\sigma_2 +\lambda_3(t)n_3\sigma_3)
\end{equation}with
\begin{equation}\nonumber
\lambda_i(t)=e^{-\Gamma_j(t)-\Gamma_k(t)}, \ \ \Gamma_i(t)=\int_{0}^t\gamma_i(\tau)d\tau,
\end{equation}where $\left(i,j,k\right)$ stands for possible cyclic permutations of $\left\{1,2,3\right\}$. Define three dynamical maps $\left\{\Lambda^{(k)}_{t}\right\}_t$ for $k=1,2,3$ by $\gamma^{(k)}_i=(1-\delta_{ki})\gamma_a+\delta_{ki}\gamma_b$ where
\begin{equation}\nonumber
\gamma_a(\tau) = 2\tau^2-6\tau+4,
\end{equation}
\begin{equation}\nonumber
\gamma_b(\tau)= \begin{cases} 0 &\mbox{if } \tau\notin \left[1,2\right] \\ 
-\gamma_a(\tau)+\epsilon & \mbox{if } \tau\in \left[1,2\right] \end{cases}. 
\end{equation}with some $0<\epsilon$. Consider a convex combination of $\left\{\Lambda^{(k)}_t\right\}_t$ with equal coefficients
\begin{equation}\label{real}
\Lambda_t(\rho)=\lambda(t)\rho+\frac{1-\lambda(t)}{2}\mathbb{I}
\end{equation}where
\begin{equation}\nonumber
\lambda(t)=\frac{1}{3}\left(\lambda_1^{(1)}(t)+2\lambda^{(1)}_2(t)\right). 
\end{equation}Observe that for
\begin{equation}\nonumber
0<\epsilon < -\ln\left(\frac{1}{2}e^{\frac{5}{3}}\left(e^{-\frac{10}{3}}+2e^{-\frac{5}{3}}-e^{-\frac{8}{3}}\right)\right)\approx 0,09
\end{equation} we get $\lambda(1)<\lambda(2)$ so that $\left\{\Lambda_t\right\}_t$ cannot be elementary, but by construction it belongs to the set of convex combinations of elementary dynamical maps $\left\{\Lambda^{(k)}_t\right\}_t$.

\end{example}

Dynamical map (\ref{real}) can be understood in the terms of choosing single elementary dynamics according to some probability distribution - one cannot restrict notion of dynamical maps without quantum memory only to elementary map.
\begin{definition}\label{qm}
Consider a dynamical map $\left\{\Lambda_t\right\}_t$ acting between $M_d(\mathbb{C})$. We say that there is no quantum information backlflow (in a strong sense) on $\left[t_1,t_2\right]$ if 
\begin{equation}\nonumber
\Lambda_t=\sum_ip_i \Lambda_t^i
\end{equation}is a time-independent convex combination where $t\in \left[t_1,t_2\right]$ and each $\Lambda^i_t$ is elementary with respect to some basis $B_i$ in $\mathbb{C}^d$. 
\end{definition} Notice that in particular dynamical maps which are RHP-Markovian and BPL-Markovian satisfy the above definition as in that case there is no backflow at all (see schematic figure \ref{fig2}). Starting form subfamilies of elementary maps, one can consider subclasses of dynamical maps without quantum information backflow (see schematic figure \ref{fig2}).

\begin{definition}\label{qm_0}
Consider a dynamical map $\left\{\Lambda_t\right\}_t$ acting between $M_d(\mathbb{C})$ that can be expressed as a time-independent convex combination $\Lambda_t=\sum_ip_i \Lambda_t^i$ where $t\in \left[t_1,t_2\right]$. If each $\Lambda^i_t$ is generalized classical (\ref{map_2}) with respect to some basis $B_i$, we say that there $\left\{\Lambda_t\right\}_t$ is of type 0 on $\left[t_1,t_2\right]$. If 
each $\Lambda^i_t$ is block-diagonal (diagonal) elementary with respect to some basis $B_i$, we say $\left\{\Lambda_t\right\}_t$ is of type I (type II) on $\left[t_1,t_2\right]$
\end{definition}

\begin{figure}[H]
\includegraphics[width=0.45\textwidth]{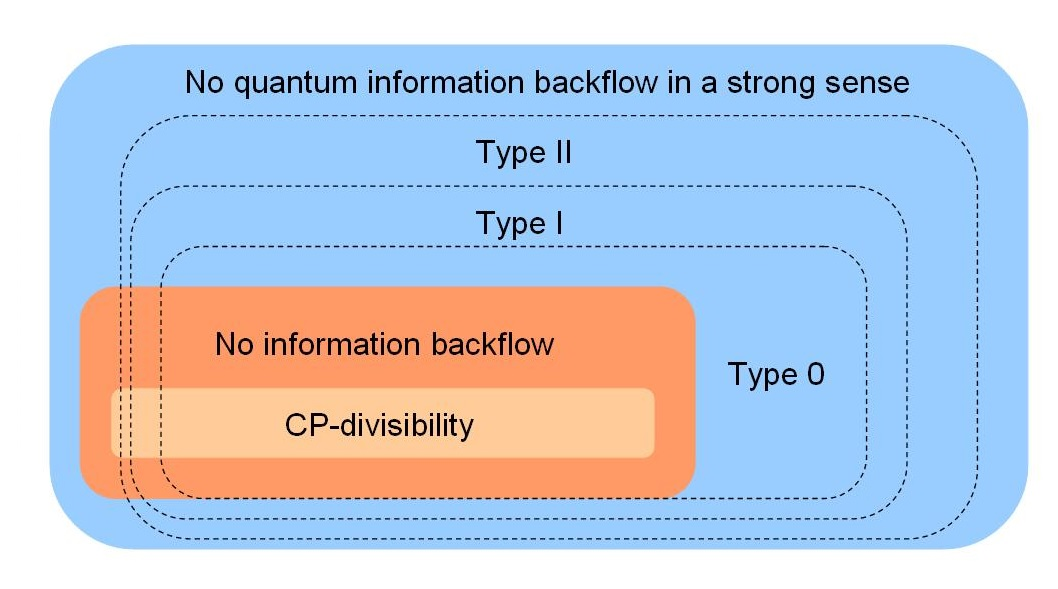}
\caption{\label{fig2}Relations between notions of lack of memory.}
\end{figure}

\begin{remark}\label{rm} We will say that dynamical map has quantum information backflow in a weak sense if it has information backflow (it is not BLP-Markovian) and is not of type 0 - in that case backflow may be simulated by probabilistic choice of classical evolution.
\end{remark}

In order to show that not all quantum processes have no quantum memory according to proposed definition \ref{qm}, consider a dynamical map $\left\{\Lambda_t\right\}_t$ such that both $\Lambda_w$ and $\Lambda_s$ are extremal (in the set of all CPTP maps) for some $w>s$. Now let us assume that for any basis $B$, there are states $\rho_1,\rho_2$ indistinguishable with respect to $B$ such that 
\begin{equation}\label{condition}
\left\|\Lambda_w(\rho_1-\rho_2)\right\|_1>\left\|\Lambda_s(\rho_1-\rho_2)\right\|_1
\end{equation}holds. Then $\left\{\Lambda_t\right\}_t$ has quantum information backflow (in a strong sense). Indeed, if $\Lambda_t$ can be expressed as a convex combination $\Lambda_t=\sum_{i=1}^np_i \Lambda_t^i$, then for $t=w,s$ we get $\Lambda_t^i=\Lambda_t$ for any $i=1,2,\ldots, n$ and by (\ref{condition}) each $\left\{\Lambda^i_t\right\}_t$ cannot be elementary.

\begin{example}\label{no_quantum}
 Consider a qubit dynamical map $\left\{\Lambda_t\right\}_t$ given by $\Lambda_t(\rho)=\frac{1}{2}(\mathbb{I}+\lambda_1(t)n_1\sigma_1+\lambda_2(t)n_2\sigma_2 +\left(\lambda_3(t)n_3+r_3(t)\right)\sigma_3)$ with $\lambda_1(t)=\lambda_2(t)=\frac{1}{2}+\frac{1}{4}\sin t$, $\lambda_3(t)=\lambda_1(t)\lambda_2(t)$ and $r_3(t)=1-\lambda_3(t)$. Observe that for any vector $\vec{m}$ and for any $t,s\in \left(0,\frac{\pi}{2}\right)$ such that $t>s$ we get $\left|T(w)\vec{m}\right|>\left|T(s)\vec{m}\right|$. Moreover, by description of extreme points of the set of CPTP maps (presented in \cite{RSW02}) one can see that for any $t\in \left(0,\frac{\pi}{2}\right)$ map $\Lambda_t$ is an extreme point, thus by previous discussion there is quantum information backflow in a strong sense.
\end{example}

Finally, let a dynamical map $\left\{\Lambda^0_t\right\}_t$ be like above (i.e. let (\ref{condition}) be fulfilled). For $i=1,2,\ldots, n$ consider any dynamical maps $\left\{\Lambda^i_t\right\}_t$. Then a time-dependent convex combination $\Lambda_t=\sum_{i=0}^np_i(t) \Lambda_t^i$ such that $p_0(s)=p_0(w)=1$ has quantum information backflow (in a strong sense) as well.\newline

Above remarks show that there is plenty room for examples of processes with truly quantum memory. Still the main remaining task is to find an applicable procedure which can certify that a given dynamical map cannot be decompose as a convex combinations of dynamical maps which are elementary (on a considered interval). This general problem seems to be difficult. However, one can in principle check that a given dynamical map is not of type 0.

Let $W\in M_d(\mathbb{C})\otimes M_d(\mathbb{C})$ be a Hermitian operator. We say that $W$ is c-c channel states witness if $\mathrm{Tr}(W\rho)\geq 0$ for any state $\rho$ of the form (\ref{c-c state}). State $\rho$ is in a set of convex combinations of c-c channel states if and only if $\mathrm{Tr}(W\rho)\geq 0$ for all c-c channel states witnesses $W$.

\begin{proposition}
Let $\left\{\Lambda_t\right\}_t$ be a dynamical map. If there is $t\in[t_1,t_2]$ such that Choi-Jamio{\l}kowski state related to $\Lambda_t$ is not given by a convex combinations of c-c channel states, then $\left\{\Lambda_t\right\}_t$ is not of type 0 on $[t_1,t_2]$. 
\end{proposition}

For two-qubit system put
\begin{equation}\nonumber
 [ 1, \vec{r},\vec{s},T]
   :=\frac{1}{4}(\mathbb{I}\otimes\mathbb{I}+ \vec{r}\vec{\sigma}\otimes \mathbb{I}+\mathbb{I}\otimes \vec{s}\vec{\sigma}+\sum_{i,j}T_{ij} \sigma_i\otimes \sigma_j). 
\end{equation}
\begin{proposition}\label{lastprop}
Any qubit c-c channel states witness is of the form $W=[ 1, 0,\vec{s}_w,T_w]$ with $|\vec{s}_w| \leq 1$ and $||T_w||_{\infty} \leq 1$. If there is $t\in[t_1,t_2]$ such that Choi-Jamio{\l}kowski state $\rho_{\Lambda_{t}}=[1, 0,\vec{s},T]$ related to $\Lambda_t$ satisfies $X(\rho_{\Lambda_{t}})=|\vec{s}| + \left\|T\right\|_{1}>1$, then $\left\{\Lambda_t\right\}_t$ is not of type 0 on $[t_1,t_2]$. Moreover, optimal witness for $\rho_{\Lambda_{t}}$ is given by $W_{op}=[1,0, - \hat{s}, - O(T)]$  where $\hat{s}$ is normalized $\vec{s}$ and $O(T)$ is an orthogonal matrix from the polar decomposition of $T$.
\end{proposition}

Note that proposition \ref{lastprop} gives operational witness of quantum information backflow (in a weak sense - see remark \ref{rm}). For detailed discussion of this propositions see Appendix.

\begin{example}
Consider once more dynamical map $\left\{\Lambda_t\right\}_t$ described in example \ref{no_quantum}. In particular by the previous discussion it is not BLP-Markovian (there is information backflow in a usual sense). Observe that due to $E_{11}=\frac{1}{2}(\mathbb{I}+\sigma_3),E_{12}=\frac{1}{2}(\sigma_1+i\sigma_2)$ and $E_{22}=\frac{1}{2}(\mathbb{I}-\sigma_3)$ the Choi-Jamio{\l}kowski state related to $\Lambda_t$ is given by 
\begin{eqnarray}
\rho_{\Lambda_{t}}&=&\frac{1}{2}\sum_{i,j=1}^2 E_{ij}\otimes \Lambda_t(E_{ij})=[1,0,\vec{s}_{\rho_{\Lambda_{t}}}(t),T_{\rho_{\Lambda_{t}}}(t)] \nonumber
\end{eqnarray}where $T_{\rho_{\Lambda_{t}}}(t)=\mathrm{diag}( \lambda_1(t),-\lambda_2(t),\lambda_3(t))$ and $\vec{s}_{\rho_{\Lambda_{t}}}(t)=(0,0,r_3(t))^T$. Calculations of length of $\vec{s}_{\rho_{\Lambda_{t}}}(t)$ and trace norm of $T_{\rho_{\Lambda_{t}}}(t)$ according to proposition \ref{lastprop} gives us
\begin{equation}\nonumber
X(\rho_{\Lambda_{t}})=r_3(t)+\lambda_1(t)+\lambda_2(t)+\lambda_3(t)=2+\frac{1}{2}\sin t>1,
\end{equation}so there is quantum information backflow in a weak sense (observed backflow cannot be simulated by convex combinations of classical dynamics).
\end{example}

\textit{Discussion.-} We have shown that the presence of information backflow should not immediately indicate quantum memory effect. In order to formalized this possibility of classical backflow, we have introduced the notion of elementary dynamical map (with respect to a given basis) and we have characterize it in the the simplest qubit case. Moreover, we have considered dynamical maps which can be expressed as a time-independent convex combinations of this elementary maps and we have used them to formulate definitions of a truly quantum information backflow.

Presented results provide a starting point and require further analysis. It is natural to ask for a characterization of those dynamical maps which can be expressed as a convex combinations of elementary dynamical maps or at least for some sufficient or necessary criteria. This task may be quite difficult in full generality, as property of being elementary with respect to some basis is not automatically preserved by convex combinations. Other interesting problems should be related to operational characterization of elementary dynamical maps for general qudit case with $d>2$.

\textit{Acknowledgments.-} The work is part of the ICTQT IRAP project of FNP. The "International Centre for Theory of Quantum Technologies" project (contract no. 2018/MAB/5) is carried out within the International Research Agendas Programme of the Foundation for Polish Science co-financed by the European Union from the funds of the Smart Growth Operational Programme, axis IV: Increasing the research potential (Measure 4.3).

 Discussions with Bihalan Bhattacharya, Samyadeb Bhattacharya, Micha{\l} Horodecki, Ryszard Horodecki and Kamil Korzekwa are acknowledged.

\appendix
\section{Witnesses of quantum information backflow - characterization of convex set of c-c channel states in a qubit case}

Notice that if a dynamical map $\left\{\Lambda_t\right\}_t$ is a convex combination of generalized classical dynamical maps (\ref{map_2}) on $[t_1,t_2]$ (is of type 0 on $[t_1,t_2]$), then for any fixed time $t\in [t_1,t_2]$ image of a CPTP map $\Lambda_t$ under Choi-Jamio{\l}kowski isomorphism is given by a convex combinations of c-c channel states
\begin{equation}\label{c-c state}
\rho_{cc}=\sum_{i,j=1}^d p_{ij}\left|f_i\right\rangle\left\langle f_i\right|\otimes \left|e_j\right\rangle\left\langle e_j\right|
\end{equation}where $\{p_{ij}\}$ satisfying $\sum_{j} p_{ij}=\frac{1}{d}$ (we use here convention of maximally mixed first system) and $\left\{\left|f_i\right\rangle\right\}_i,\left\{\left|e_j\right\rangle\right\}_j$ are some fixed orthonormal basis. Indeed channel state for map (\ref{map_2}) is given by
\begin{eqnarray}
  \rho_{gcl} &=& \frac{1}{d}\sum_{k,l,i,j=1}^d \Lambda_{ij}(t)\left\langle e_j\right|E_{kl}\left|e_j\right\rangle E_{kl}\otimes\left|f_i\right\rangle\left\langle f_i\right| \nonumber\\
   &=& \sum_{i,j=1}^d \frac{\Lambda_{ij}(t)}{d}\left|\tilde{e}_j\right\rangle\left\langle \tilde{e}_j\right|\otimes\left|f_i\right\rangle\left\langle f_i\right| \nonumber
\end{eqnarray}where for any $i=1,2,\ldots ,d$ we put $\left|f_i\right\rangle=U(t)\left|e_i\right\rangle$ and $\left|\tilde{e}_j\right\rangle$ stands for conjugated vector (with respect to the computational basis such that $E_{kl}=\left|k\right\rangle\left\langle l\right|$). Obviously $\rho_{gcl}$ is of the form (\ref{c-c state}) as $\Lambda(t)$ is a stochastic matrix.

Therefore, in order to show that a given evolution cannot be represented as a convex combination of generalized classical dynamical maps it is enough to find time $t$ for which the Choi-Jamio{\l}kowski state related to the given dynamical map is not a convex combination of c-c channel states. This can be done using witnesses $W$ for c-c channel states.

 Observe that for any bases $\left\{\left|f_i\right\rangle\right\}_i, \left\{\left|e_j\right\rangle\right\}_j$ there are local unitaries $U,V\in M_d(\mathbb{C})$ responsible for changing basis to the computational one. Put $W^{U,V}=U \otimes V W  U^{\dagger} \otimes V^{\dagger} $ and $W^{U,V}_{ij}= [U \otimes V W  U^{\dagger} \otimes V^{\dagger} ]_{ii,jj}$. Then the condition for being c-c channel states witness is given by
\begin{equation}\label{characterisation}
\sum_{i,j=1}^dp_{ij} W^{U,V}_{ij} \geq 0
\end{equation}for any pair  $U,V$ and all c-c channels probabilites. The extremal points of the set of such 
probabilities are all $\{p_{ij}\}$ such that $ p_{ij}=\delta_{i,j(i)} \frac{1}{d}$ for $ j(i)$  being some deterministic function of index i. 
- it is enough to satisfy condition (\ref{characterisation}) only for this type of probabilities.

Now let us restrict our attention to the qubit case. The representation of normalized witness ($\mathrm{Tr} W = 1$) in the 
Hilbert-Schmidt basis is given as
\begin{eqnarray}
  W = [ 1, \vec{r},\vec{s},T]  \label{representation}
\end{eqnarray}with real vectors $\vec{r},\vec{s} \in R^{3}$ and real matrix $T\in M_3(\mathbb{R})$. Condition (\ref{characterisation}) boils down to
\begin{eqnarray}
W^{U,V}_{11} + W^{U,V}_{21} \geq 0;\ W^{U,V}_{12} + W^{U,V}_{22} \geq 0; \nonumber \\
W^{U,V}_{11} + W^{U,V}_{22} \geq 0; \ W^{U,V}_{12} + W^{U,V}_{21}\geq 0\nonumber
\label{qubit-cond}.
\end{eqnarray}With a little bit of calculations one can show that this is equivalent to $\mathrm{Tr}\left(W^{U,V}Z\right)\geq 0$ for all $Z$ of the following form
\begin{eqnarray}\nonumber
[ 1, 0, \hat{z}, 0 ] ;\  [ 1, 0,  - \hat{z}, 0 ] ; \ [ 1, 0, 0, | \hat{z} \rangle \langle \hat{z}|  ]; \ 
[1, 0, 0, - | \hat{z} \rangle \langle \hat{z}| ].
\end{eqnarray} 
Transferring action of $U,V$ form $W$ to $Z$ we see that the previous condition is true if and only if $\mathrm{Tr}\left(WZ\right)\geq 0$ with all $Z$ of the form
\begin{eqnarray}\nonumber
[ 1, 0, \hat{a}, 0 ] ;\  [ 1, 0,  - \hat{a}, 0 ] ; \ [ 1, 0, 0, | \hat{a} \rangle \langle \hat{b}|  ]; \ 
[1, 0, 0, - | \hat{a} \rangle \langle \hat{b}| ] 
\end{eqnarray} 
for all normalized vectors $\hat{a}$ and $\hat{b}$. Therefore, the conditions (\ref{characterisation}) are given by
\begin{eqnarray}\nonumber
1 \pm \vec{s}\hat{a} \geq 0;\ 1 \pm \langle \hat{a} |T | \hat{b} \rangle \geq 0\nonumber
\label{qubit-cond1}.
\end{eqnarray}If so, then qubit c-c channel states witnesses are all the operators (\ref{representation}) with the parameters satisfying 
\begin{equation}
|\vec{s}| \leq 1, \ ||T||_{\infty} \leq 1
\label{qubit-witn-characterisation}
\end{equation}and the arbitrary vector $\vec{r}$ (which may be put by convention to zero). Consider now any channel state $\rho=[1,0,\vec{s}_{\rho},T_{\rho}]$. It is a convex combination of c-c channel states if and only if a minimum of $\mathrm{Tr}(W \rho)$ over 
all witnesses satisfying (\ref{qubit-witn-characterisation}) is non-negative. This can be equivalently expressed by condition
\begin{equation}\nonumber
1 +\mathrm{min}_{|\vec{s}|\leq 1}( \vec{s}_{\rho},\vec{s} ) + \mathrm{min}_{||T||_{\infty} \leq 1}
 \mathrm{Tr}\left(T_{\rho}T^{T}\right) \geq 0.
\end{equation}
Note that optimizations over $\vec{s}$ and $T$ are independent. Therefore, $\rho$ is a convex combination of c-c channel states if and only if 
\begin{equation}\label{final}
X(\rho) := |\vec{s}_{\rho}| + \left\|T_{\rho}\right\|_{1} \leq  1.
\end{equation}

The optimal c-c channel states witness for a given channel state $\rho$ is constructed as  $W_op(\rho)=[1,0, - \hat{s}_{\rho}, - O(T_{\rho})]$ where $\hat{s}_{\rho}= \frac{\vec{s}_{\rho}}{|\vec{s}_{\rho} |}$ and $O(T_\rho)$ is an orthogonal matrix from the polar decomposition of $T_\rho$.

If there is a $t\in [t_1,t_2]$ such that $\rho_{\Lambda_{t}}$ expressed via (\ref{representation}) violates (\ref{final}), then $\left\{\Lambda_t\right\}_t$ is not of type 0 on $[t_1,t_2]$. Therefore, $X(\rho)$ can be seen as a witness of quantum information backflow in a weak sense.\newline
\end{document}